\documentclass{article}                    
\usepackage{vmargin}
\setmarginsrb{1.1in}{1.1in}{1.1in}{1.1in}{0pt}{5mm}{0pt}{7mm}
\usepackage[noadjust]{cite} 
\usepackage{graphicx}
\usepackage{amssymb,amsmath,textcomp}
\usepackage{amsthm}
\usepackage{color}
\usepackage{hyperref}
\usepackage{tabularx}
\usepackage[justification=centering]{caption}
\usepackage{latexsym}
\usepackage{authblk}
\usepackage{enumerate}
\usepackage[utf8]{inputenc}
\usepackage{comment}

\newcommand{\remove}[1]{}

\newcommand{\qExpLem}{{\sf $q$-Expansion Lemma}}

\newcommand{\defparproblem}[4]{
 \vspace{3mm}
\noindent\fbox{
  \begin{minipage}{.95\textwidth}
  \begin{tabular*}{\textwidth}{@{\extracolsep{\fill}}lr} \textsc{#1} \\ \end{tabular*}
  {\bf{Input:}} #2  \\
  {\bf{Parameter:}} #3\\
  {\bf{Question:}} #4
  \end{minipage}
  }
  \vspace{2mm}
}

\newcommand{\defproblem}[3]{
  \vspace{3mm}
\noindent\fbox{
  \begin{minipage}{.95\textwidth}
  \begin{tabular*}{\textwidth}{@{\extracolsep{\fill}}lr} #1  \\ \end{tabular*}
  {\bf{Input:}} #2  \\
  {\bf{Question:}} #3
  \end{minipage}
  }
  \vspace{2mm}
  }

\usepackage{microtype}

\newtheorem{definition}{\bf Definition}[section]

\newtheorem{reduction rule}{\bf Reduction Rule}[section]
\newtheorem{proposition}{\bf Proposition}[section]
\newtheorem{lemma}{\bf Lemma}[section]

\newcommand{\ECT}{{\sc Even Cycle Transversal}}
\newcommand{\CVD}{{\sc Cluster Vertex Deletion}}
\newcommand{\VC}{{\sc Vertex Cover}}
\newcommand{\PWOneVD}{{\sc Pathwidth One Vertex Deletion}}
\newcommand{\FVS}{{\sc Feedback Vertex Set}}
\newcommand{\ChordVD}{{\sc Chordal Vertex Deletion}}
\newcommand{\BlockVD}{{\sc Block Graph Vertex Deletion}}
\newcommand{\ThreePathPacking}{{\sc $3$-Path Packing}}
\newcommand{\TPT}{{\sc Triangle Packing in Tournaments}}
\newcommand{\FVST}{{\sc Feedback Vertex Set in Tournaments}}
\newcommand{\OFVDS}{{\sc Out-Forest Vertex Deletion Set}}

\newcommand{\dPVC}{{\sc $d$-Path Vertex Cover}}
\newcommand{\tPDCP}{{\sc Pairwise $t$-Disjoint Cycle Packing}}
\newcommand{\tADCP}{{\sc Almost $t$-Disjoint Cycle Packing}}
\newcommand{\IFVS}{{\sc Independent Feedback Vertex Set}}

\newcommand{\WSeparation}{{\sc $W$-Weight Separator}}
\newcommand{\WPacking}{{\sc $W$-Weight Packing}}
\newcommand{\BCP}{{\sc Max-Min (Min-Max) Balanced Connected Partition}}

\newcommand{\ListColor}{{\sc List Coloring}}
\newcommand{\ArcRCP}{{\sc Arc Disjoint $r$-Cycle Packing}}
\newcommand{\ArcFourCP}{{\sc Arc Disjoint $4$-Cycle Packing}}

\newcommand{\CC}{{\mathcal C}}

\newcommand{\FF}{\ensuremath{\mathcal{F}}\xspace}

\newcommand{\HH}{{\mathcal H}}

\newcommand{\NN}{{\mathcal N}}
\newcommand{\OO}{\mathcal{O}}
\newcommand{\PP}{\ensuremath{\mathcal{P}}\xspace}

\newcommand{\UU}{{\mathcal U}}

\newcommand{\XX}{{\mathcal X}}

\newcommand{\nn}{{\mathbb N}}

\newcommand{\cc}{{\mathbb C}}

\newcommand{\fR}{{\mathfrak R}}

\newcommand{\nka}{${\sf NP \subseteq coNP/poly}$}

\usepackage{tikz}
\usetikzlibrary{decorations.shapes,decorations.pathreplacing,decorations.pathmorphing}
\usetikzlibrary{arrows,matrix,shapes}
\usetikzlibrary{positioning}
\tikzset{
        stars/.style={star,inner sep=2pt}
    }

\usepackage{todonotes}
\usepackage[ruled,vlined,linesnumbered]{algorithm2e}

\title{Expansion Lemma - Variations and Applications to Polynomial-time Preprocessing}

\author[1]{Ashwin Jacob}
\author[2]{Diptapriyo Majumdar}
\author[3]{Venkatesh Raman}
\affil[1]{Ben Gurion University of the Negev, Beersheva, Israel\\
  \texttt{ashwinj@bgu.ac.il}}
\affil[2]{Indraprastha Institute of Information Technology Delhi, New Delhi, India\\
	\texttt{diptapriyo@iiitd.ac.in}}
\affil[3]{The Institute of Mathematical Sciences, HBNI, Chennai, India\\
	\texttt{vraman@imsc.res.in}}

\begin{document}
\maketitle

\begin{abstract}
In parameterized complexity, it is well-known that a parameterized problem is fixed-parameter tractable if and only if it has a kernel - an instance equivalent to the input instance, whose size is just a function of the parameter. The size of the kernel can be exponential or worse, resulting in a quest for fixed-parameter tractable problems with a polynomial-sized kernel. The developments in machinery to show lower bounds for the sizes of the kernel gave rise to the question of the asymptotically optimum size for the kernel of fixed-parameter tractable problems.

In this article, we survey a tool called expansion lemma that helps in reducing the size of the kernel. Its early origin is in the form of Crown Decomposition for obtaining linear kernel for the {\VC} problem and the specific lemma was identified as the tool behind the optimal $O(k^2)$ kernel for the undirected feedback vertex set problem. Since then, several variations and extensions of the tool have been discovered. We survey them along with their applications in this article. 

\end{abstract}

\newpage

\tableofcontents

\newpage

\section{Introduction}
\label{sec:intro}
Preprocessing is an age-old technique in the implementation of algorithms, particularly for hard problems.
With the advent of parameterized complexity, the theory of preprocessing got a lease of formal life. A parameterized problem comes with a parameter apart from the input instance. The goal of parameterized complexity is to see whether one can design an algorithm whose combinatorial explosion can be confined to the parameter, while the rest of the running time is polynomial in the input size. Such an algorithm is said to be a fixed-parameter algorithm, and a problem admitting such an algorithm is said to be fixed-parameter tractable.

It is well-known (see, for example, \cite{CFKLMPPS15}) that a decidable parameterized problem is fixed-parameter tractable if and only if the input instance can be reduced in polynomial time to an equivalent instance, whose size is a function of just the parameter. Such an equivalent instance is said to be a kernel.
The size of the kernel can be exponential or worse, resulting in a quest for fixed-parameter tractable problems with a polynomial-sized kernel. Machinery to show the non-existence of polynomial-sized kernels~\cite{bodlaender2009problems,dom2014kernelization}, and more specific lower bounds for kernel sizes~\cite{dell2014satisfiability} under complexity-theoretic assumptions, made the theory more interesting.  

In this article, we survey a tool called expansion lemma that helps in reducing the size of the kernel. Its early origin is in the form of Crown Decomposition for obtaining linear kernel for the {\VC} problem \cite{chor2004linear}, and the specific lemma was identified as the tool behind the optimal $O(k^2)$ kernel for the {\sc Undirected {\FVS}} problem \cite{Thomasse2010}. Since then, several variations and extensions of the tool have been discovered. In this paper, we survey them along with their applications.
We have tried to provide sketches to as many results as possible wherever the main (and probably only) tool for the application is the expansion lemma variation.
For several problems, there are several other ideas involved in obtaining the improved kernel beyond the tool.
Also, wherever the main tool is expansion lemaa, we have tried explain the algorithm in a self-contained way, we have done that.
For others, we have tried to indicate how the tool is applied, and often we didn’t go into the other complicated technical details as they distract the main purpose of the survey.

\noindent
{\bf Organisation of the Paper:} In Section \ref{sec:prelim}, we introduce basic terminologies and notations.
We begin the technical part of the paper in the following three sections.
In Section \ref{sec:crown}, we introduce the crown decomposition technique, basic lemmas, and their application to obtain some improved kernels.
In Section \ref{sec:expansion-lemma-basic}, we introduce the expansion lemma and its use for several problems in kernelization.
Finally, in Section \ref{sec:new-double-expansion-lemma}, we discuss the recent developments, including generalizations of the (basic) expansion lemma, along with their applications.
We conclude with some possible future research directions in Section \ref{sec:conclusion}.

\section{Preliminaries}
\label{sec:prelim}

\paragraph{Parameterized Complexity and Kernelization:}
A parameterized problem $L$ is a subset of $\Sigma^* \times \nn$ where $\Sigma$ is a finite alphabet.
It is assumed that $k \in \nn$ is given in unary and an instance to a parameterized problem is $(x, k)$ where $x \in \Sigma^*$ and $k \in \nn$.

\begin{definition}[Fixed-Parameter Tractability]
\label{defn:FPT}
A parameterized problem $L \subseteq \Sigma^* \times \nn$ is said to be {\em fixed-parameter tractable} (or {\em FPT}) if there exists an algorithm for solving this problem $L$ that on input $(x, k)$, runs in $f(k)|x|^{\OO(1)}$-time where $f: \nn \rightarrow \nn$ is a computable function.
\end{definition}

\begin{definition}[Kernelization]
\label{defn:kernelization}
Let $L \subseteq \Sigma^* \times \nn$ be a parameterized problem.
A {\em kernelization} for $L$ is a polynomial time ($(|x| + k)^{\OO(1)}$-time) procedure that replaces an input instance $(x, k)$ by an input instance $(x',k')$ such that
\begin{itemize}
	\item $(x,k) \in L$ if and only if $(x',k') \in L$.
	\item $|x'| + k' \leq g(k)$ for some computable function $g: \nn \rightarrow \nn$.
\end{itemize}	
\end{definition}

It is well-known that a decidable parameterized problem is FPT if and only if it admits a kernelization~\cite{CFKLMPPS15}.
Kernelization can be viewed as a theoretical foundation of {\em preprocessing heuristics}.
We are interested in investigating when a parameterized problem admits a kernelization of small size, particularly kernels of polynomial size with as small an upper bound as possible.
For other definitions, notations, and details on parameterized complexity and kernelization, 
please refer to \cite{CFKLMPPS15,Niedermeier06,GuoN07,DowneyF13,FlumG06} and \cite{FLSZ19}.

\paragraph{Sets, Numbers, and Graph Theory:} We use $\nn$ to denote the set of natural numbers. For $r \in \nn$, we denote the set $\{1,\ldots,r\}$ by $[r]$.
A {\em boolean formula} is an expression of binary variables with logical operators conjunction ($\wedge$), disjunction ($\lor$), and negation $(\neg$). 
A {boolean formula} is in {\em conjunctive normal form} (or {\em CNF}) if it is a conjuntion of clauses.
A {\em literal} is a boolean variable either in its own form or in its complemented form.
Each clause is a disjunction of literal.
We usually use $G = (V, E)$ to denote an undirected graph, and $D = (V, A)$ to denote a directed graph. Let $V(G)$ denote the set of vertices of $G$, and $E(G)$ the set of edges of $G$. When $G$ is clear from context, let $n = |V(G)|$ and $m = |E(G)|$.
Informally, we use $n$ to denote the number of vertices and $m$ to denote the number of edges in the input graph.
Given a graph $G = (V, E)$ and $X \subseteq V(G)$, we use $G[X]$ to denote an {\em induced subgraph} of $G$; its vertex set is $X$, and its edge set consists of all the edges of $E(G)$ that have both endpoints in $X$. Let $G - X$ denote the subgraph of $G$ induced on the vertex set $V(G) \setminus X$.
An undirected graph is called a {\em clique} if every pair of vertices are adjacent to each other.
An {\em independent set} is a pairwise non-adjacent set of vertices. 
A {\em tournament} is a directed graph $T$ such that for every pair of vertices $u,v \in V(T)$, exactly one of $uv$ or $vu$ is a directed edge of $T$.
In an undirected graph $G = (V, E)$, a {\em path} is a set $P$ of distinct vertices (except possibly for the start vertex and the end vertex) such that the vertices of $P$ can be arranged in an order $v_1,\ldots,v_d$ such that for every $1 \leq i \leq d-1$, $v_i v_{i+1} \in E(G)$.
Informally, for all $1 \leq i \leq d-1$, $v_i \neq v_{i-1}$ for a path.
Similarly, in a directed graph $D = (V, A)$, a {\em directed path} can be defined in a similar way when for every $i \leq d-1$, $(v_i, v_{i+1}) \in A$.
A (directed) {\em cycle} is a (directed) path that starts and ends at the same vertex. A {\em connected component} of an undirected graph is a connected subgraph that is not part of any larger connected subgraph. For subsets $X,A,B \subseteq V(G)$, we say $X$ {\em separates} $A$ and $B$ when no component of $G - X$ contains vertices from both $A \setminus X$ and $B \setminus X$.
Given a vertex subset $X \subseteq V(G)$, we use $N_G(X)$ to denote the {\em open neighborhood} of $X$ and $N_G[X]$ to denote the {\em closed neighborhood} of $X$.
When graph $G$ is clear from the context, we omit the subscript.

A vertex $v$ of $G$ is a {\em cut vertex } if $G - \{v\}$ has more connected components than $G$. A
graph $G$ is {\em biconnected} if $|V(G) \geq 3$, and it has no cut vertex.
A graph $G$ is a {\em block graph} if each of its maximal biconnected components is a clique.
A {\em tree} is a connected undirected graph with no cycle (or equivalently, a connected undirected acyclic graph).
A {\em forest} is a graph, every connected component of which is a tree.
A {\em directed acyclic graph} is a directed graph with no directed cycles. A {\em topological ordering} for a directed acyclic graph is a linear ordering of vertices such that for every directed edge $uv$, vertex $u$ comes before $v$ in the ordering.
An undirected graph $G = (V, E)$ is said to be {\em bipartite} if $V(G) = A \uplus B$ such that for every $uv \in E(G)$, $u \in A$ if and only if $v \in B$.
We call $A$ and $B$ the bipartitions of $G$.
Another alternate characterization is that a bipartite graph has no odd cycle.
We use $G = (A \uplus B, E)$ to denote a bipartite graph such that $A$ and $B$ are its bipartitions and $E$ is the set of its edges.
A set $M \subseteq E(G)$ is said to be a {\em matching} of a graph if for every two edges $uv, xy \in M$, $\{u, v\} \cap \{x, y\} = \emptyset$.
A vertex $u \in V(G)$ is said to be {\em matched} by a matching $M$, if there is an edge of $M$ incident on $u$.
A matching $M$ of $G$ is said to {\em saturate} $A$ if $A \subseteq V(M)$.
A matching is said to be {\em perfect matching} if all vertices are matched by $M$.
A cycle $C$ in a graph is said to be an {\em induced cycle} if $G[C]$ has no edge other than the edges of the cycle.
A graph is said to be a {\em chordal graph} if it has no induced cycle of length at least four.
An {\em out-tree} is a directed graph where each vertex has in-degree at most 1, and the underlying (undirected) graph is a tree. An {\em out-forest} is a disjoint union of out-trees.
A graph is said to be {\em $k$-colorable} if there is a function $\lambda: V(G) \rightarrow [k]$ for some $k \in \nn$ such that for every $uv \in E(G)$, $\lambda(u) \neq \lambda(v)$.
Such a function $\lambda$ is called a {\em proper coloring} of $G$.
It is well-known that a bipartite graph is $2$-colorable.
A graph $G$ is said to be {\em factor-critical} if for every $u \in V(G)$, the graph $G - u$ has a perfect matching.
For other terminologies of graph theory, we use standard graph-theoretic notations by Diestel \cite{Diestel2012}.

We use the following well-known proposition 
 several times in the paper.

\begin{proposition}[Hall's Theorem \cite{Hall1935}]
\label{prop:hall-theorem}
Let $G = (V = A \uplus B, E)$ be a bipartite graph.
Then, $G$ has a matching saturating $A$ if and only if for every $X \subseteq A$, $|N(X)| \geq |X|$.
\end{proposition}

Given a bipartite graph $G = (A \uplus B, E)$, a set $X \subseteq A$ is said to be a {\em hall-set} if $|N(X)| < |X|$. 
A hall-set is said to be {\em minimal} if none of its proper subsets is a hall-set.
Given a bipartite graph $G = (A \uplus B, E)$, a minimal hall set can be computed in polynomial-time using Proposition \ref{prop:hall-theorem}.

Before we state the following proposition, we need a few notations.
For a graph $G = (V, E)$, suppose that $S \subseteq V(G)$ and a set $\CC$ of vertex-disjoint subgraphs of $G - S$, let $\NN_G(S, \CC) = \{H \in \CC \mid E(S, V(H)) \neq \emptyset\}$. 

\begin{proposition}[Edmond's Gallai Structure Theorem \cite{edmonds_1965,LovaszP2000}]
\label{prop:edmond-gallai}
For every graph $G = (V, E)$, there are disjoint subsets $X, Y, Z \subseteq V(G)$ such that 
\begin{enumerate}[(i)]
	\item $X \cup Y \cup Z = V(G)$,
	\item $Y = N_G(X)$,
	\item every connected component of $G[X]$ is factor-critical,
	\item $G[Z]$ has a perfect matching, and
	\item for every nonempty subset $Y' \subseteq Y$, it holds that $|\NN_G(Y', \CC)| \geq |Y'|$.
\end{enumerate}
Furthermore, such $X, Y, Z$ can be computed in polynomial-time.
\end{proposition}

\paragraph{Graph Parameters:} We introduce different graph parameters that we will use in this paper.
A subset $S \subseteq V(G)$ is a {\em vertex cover} of $G$ if for every edge $uv \in E(G)$, $\{u, v\} \cap S \neq \emptyset$.
It follows that for an undirected graph $G$, if $S$ is a vertex cover, then $V(G) \setminus S$ is an independent set (and vice versa).
A subset $S \subseteq V(G)$ is called a {\em cluster vertex deletion set} of $G$ if every connected component of $G - S$ is a clique.
A subset $S \subseteq V(G)$ is a {\em feedback vertex set} of $G$ if $G - S$ has no cycle.

The following graph parameter requires a more non-trivial definition.


\begin{definition}[Pathwidth \cite{RobertsonS84,RobertsonS86,RobertsonS90a}]
\label{defn:pathwidth}	
A {\em path decomposition} of a graph $G$ is a pair $(T, \chi)$ in which $T = (V_T, E_T)$ is a path and $\chi = \{\chi_i \mid i \in V_T\}$ is a family of subsets of $V(G)$ called {\em bags}, such that
\begin{enumerate}[(i)]
	\item\label{vertex-belong-tw} $\cup_{i \in V_T} \chi_i = V(G)$.
	\item\label{edge-belong-tw} for each edge $uv \in E(G)$, there is $i \in V_T$ such that $u, v \in \chi_i$.
	\item\label{connectivity-tw} for each vertex $u \in V(G)$, the set of nodes $\{i \in V_T \mid u \in \chi_i\}$ induces a connected subgraph of $T$.
\end{enumerate}
\end{definition}

The value $\max_{i \in V_T} \{|\chi_i| - 1\}$ is called the {\em width} of the path decomposition $(T, \chi)$.
The {\em pathwidth} of a graph is the minimum width taken over all possible tree decompositions of $G$.
It is not very hard to understand that the pathwidth of a path (or a linear forest) is 1.
Similarly, the pathwidth of a cycle is 2, and the pathwidth of a caterpillar is also 1.
Also, the pathwidth of a clique with $n$ vertices is $n-1$.
On an informal note, the pathwidth of a graph is a measure on how close a graph is to a path.



We begin the technical part of the survey with a structure and a lemma that could be considered a precursor to expansion lemma.

\section{Crown Decomposition and Applications}
\label{sec:crown}

In this section, we define a related technique called ``crown decomposition'' and its application to some problems.
In Sections \ref{sec:vc-crown}, \ref{sec:(n-k)coloring}, \ref{sec:max-sat}, \ref{sec:list-color}, and \ref{sec:long-cycle}, we discuss how crown decomposition techniques are used to get kernels for {\VC}, {\sc $(n-k)$-Coloring}, {\sc Maximum Satisfiability}, {\sc $(n-k)$-List Coloring} and {\sc Longest Cycle} problems respectively.
After that, in Section \ref{sec:gen-crown}, we discuss a generalization of the crown decomposition technique recently introduced by Chen et al. \cite{ChenFSWY19}.

\defparproblem{{\VC} (VC)}{An undirected graph $G = (V, E)$ and an integer $k$.}{$k$.}{Is there a set $S \subseteq V(G)$ of at most $k$ vertices such that for every $uv \in E(G)$, $u \in S$ or $v \in S$?}

In the context of the {\sc Vertex Cover} problem, a well-known reduction rule to get rid of pendant (degree one) vertices is to pick their neighbours (and delete them and their neighbours) in the solution~\cite{CFKLMPPS15}.

A crown introduced by Fellows \cite{fellows2003blow} generalizes the structure underlying pendant vertices.

\begin{definition}[Crown Decomposition]
\label{defn:crown-decomposition}
Given a graph $G = (V, E)$, a {\em crown decomposition} of $G$ is a partition of $V(G) = C \uplus H \uplus R$ such that
\begin{itemize}
	\item $C \neq \emptyset$,
	\item $C$ is an independent set,
	\item there is no edge between a vertex of $C$ and a vertex of $R$, i.e. $H$ separates $C$ from $R$, and
	\item there is a matching of size $|H|$ among the edges $E \cap (H \times C)$.
\end{itemize}
\end{definition}

The set $C$ can be seen as a crown put on the head $H$ of the remaining part (body) $R$. If $G$ has a pendant vertex $x$ with its (unique) neighbour $y$, then $C=\{x\}$, $H =\{y\}$ and $R= V\setminus \{x, y\}$ constitute a crown decomposition.

\begin{lemma}[Crown Lemma-1,~\cite{fellows2003blow,chor2004linear}]
\label{lemma:crown-lemma1}
Let $G$ be a graph without isolated vertices with at least $3k + 1$ vertices. Then there is a polynomial-time algorithm that either
\begin{itemize}
	\item finds a matching of size $k + 1$ in $G$, or
	\item finds a crown decomposition of $G$.
\end{itemize}
\end{lemma}

Fellows \cite{fellows2003blow} apply Lemma~\ref{lemma:crown-lemma1} for 
kernelization in the following way.
First, they observe that if there is a matching of size $k+1$, then the instance is a trivial yes or no instance for some problems. Moreover, in the other case, the head of the crown can be argued to be in some optimum solution (as it hits a large number of obstructions) and hence can be deleted, resulting in a reduction rule.

The following variation of Lemma \ref{lemma:crown-lemma1} can be useful sometimes. This lemma also serves as a precursor to the more general expansion lemma described in the next section.

\begin{lemma}[Crown Lemma-2,~\cite{hopcroft1973n,chor2004linear,FLSZ19}]
\label{lemma:crown-lemma2}
Let $G = (V = A \uplus B, E)$ be a bipartite graph with no isolated vertices in $B$, and $|B| \geq |A|$. 
Then there is a polynomial time algorithm that either
\begin{itemize}
	\item finds a matching saturating $A$ or 
	\item finds a crown decomposition $(C,H,R)$ of $G$, with $C \subseteq A$ and $H \subseteq B$.
\end{itemize}
\end{lemma}

If there is no matching saturating $A$, then by Hall's Theorem (Proposition \ref{prop:hall-theorem}), there exists a minimal hall-set $C \subseteq A$ such that $|N(C)| < |C|$.
Such a minimal hall-set gives rise to a crown decomposition; see \cite{FLSZ19} for proof.

We exemplify these lemmas using some examples.

\subsection{Vertex Cover}
\label{sec:vc-crown}

We now describe a summary of a kernel of $3k$ vertices for {\VC} provided by Chor et al. \cite{chor2004linear} where the earliest application of  Lemma~\ref{lemma:crown-lemma1} can be seen.
The isolated vertices of the input graph $G$ can be safely removed. Then, they invoke the algorithm in Lemma \ref{lemma:crown-lemma1}, which either gives a matching of size $k+1$ in $G$ or a crown decomposition of $G$ if $G$ has more than $3k$ vertices. In the former case, it can be concluded that the input instance is a no-instance as any vertex cover of $G$ needs at least $k+1$ vertices to cover the edges of the matching. 

Hence they assume that the algorithm returns a crown decomposition of $G$ with crown $C$ and head $H$. Let $M$ be the matching of $H$ into $C$ that saturates $H$. Any vertex cover that covers the edges in $G[H \cup C]$ has to contain at least $|M| = |H|$ many vertices. Since $M$ saturates $H$ and $C$ is an independent set, the set $H$ covers all the edges in $G[H \cup C]$. Hence we can safely reduce the problem instance to $(G \setminus (H \cup C), k - |H|)$. Hence we conclude that we can reduce the {\sc Vertex Cover} instance as long as $|V(G) | > 3k$. Hence {\sc Vertex Cover} has a kernel of $3k$ vertices. This bound has been improved to a kernel with $2k - \log k$ vertices by Lampis et al.~\cite{lampis2011kernel} using linear programming techniques.
See Li and Zhu \cite{LiZhu2018} for another $2k$-vertex kernel for {\VC}.

\subsection{$(n-k)$-Coloring}
\label{sec:(n-k)coloring}
Given an undirected graph, determining whether its vertices can be properly colored with $k$ colors is NP-hard even for $k=3$. So a kernel or an FPT algorithm is not possible unless P=NP. 
However, a related problem is whether the graph $G$ on $n$ vertices can be colored properly with $n-k$ colors, which is the same as whether $k$ colors can be saved.
The formal definition is the following.

\defparproblem{{\sc $(n-k)$-Coloring}}{An undirected graph $G = (V, E)$ and an integer $k$.}{$k$.}{Is there a coloring $\Xi :V \rightarrow \{1,2,\ldots |V|-k\}$ such that $\Xi (x) \neq \Xi (y)$ whenever $xy \in E$?}

This problem has an FPT algorithm and a $3k$ vertex kernel using Lemma~\ref{lemma:crown-lemma1} by Chor et al. \cite{chor2004linear} (also see \cite{FLSZ19} for more details).
The idea is to apply Lemma \ref{lemma:crown-lemma2} to the complement graph. If it contains a matching of size $k$, then those $k$ pairs of non-edges in $G$ can help us save $k$ colors. If not, the crown $C$ is a clique in $G$ needing $|C|$ colors, and all vertices in $C$ are adjacent to every vertex of $V\setminus (H \cup C)$. Furthermore, the matching saturating $H$ helps to color vertices in $H$ using colors used in $C$. Thus we can remove $H \cup C$ and recurse on the remaining vertices. This process stops if we discover that we have already used more than $(n-k)$ colors, so the input is a no-instance or that there are only $3k$ vertices.
Recently, Li et al. \cite{LiDYR21} have improved this result by giving a $(2+\varepsilon)k$ vertex kernel for this problem.
They have provided a new structure called {`fractal critical crown'} of a graph, use Proposition \ref{prop:edmond-gallai} and the notion of factor-critical crown (see \ref{sec:gen-crown} for definition) to get their result.

\subsection{Maximum Satisfiability}
\label{sec:max-sat}

\defparproblem{{\sc Maximum Satisfiability}}{A Conjunctive Normal Form (CNF) formula $\phi$ with $n$ variables and $m$ clauses, and a nonnegative integer $k$.}{$k$.}{Does $\phi$ have a truth assignment satisfying at least $k$ clauses?}

It is folklore that such a CNF formula has a truth assignment satisfying at least $m/2$ clauses, which implies that we can assume that $m \leq 2k$, as otherwise, the instance is a YES-instance.
It is also easy to see (by a greedy algorithm that appropriately sets an unassigned variable in each clause in sequence) that the formula $\phi$ is satisfiable if $n \geq m$.

Thus it can be assumed that $n < m < 2k$.
Lokshtanov \cite{lokshtanov2009new} showed (as below) using Lemma~\ref{lemma:crown-lemma2} that $n$ can be made to be less than $k$. The idea is to construct a bipartite graph with bipartition $A$ and $B$ where $A$ corresponds to the variables in the formula, and $B$ corresponds to the clauses in the formula. There is an edge between $x \in A$ and a $c \in B$ if $x$ appears in the clause $c$.

If $n \geq k$ then from Lemma~\ref{lemma:crown-lemma2}, if we have a matching saturating $A$, then the clauses in the matching can be satisfied by setting the variables in the matching appropriately. Thus we will be able to satisfy at least $k$ clauses. Otherwise, it can be argued that the crown $(C, H, R)$ can be used to satisfy the clauses represented by $H$ by appropriately setting the variables represented by $C$, and removing those clauses and variables.
See \cite{FLSZ19} for more details.

\subsection{$(n-k)$-List Coloring} 
\label{sec:list-color}

Banik et al.~\cite{banik2020fixed} use 
Lemma~\ref{lemma:crown-lemma2}
to obtain an {FPT} algorithm for {\sc $(n-k)$-{\ListColor}}, a generalization of the problem in Section~\ref{sec:(n-k)coloring} where they gave a reduction rule that bound the number of colors over all the lists to $n$.
The problem is defined formally below.

\defparproblem{$(n-k)$-{\ListColor}}{A graph $G = (V, E)$, an integer $k$ and for every $u \in V(G)$, a set $L(u)$ of exactly $(n-k)$ colors.}{$k$.}{Is there a proper coloring of $G$ such that for all $u \in V(G)$ the color assigned to $u$ is in $L(u)$?}

Here they create an auxiliary bipartite graph with vertices $V(G)$ on one side and $\mathcal{C}$, the colors in every vertex lists on the other side, adding edges between a vertex and a color if its list contains it. After that, they invoke Lemma~\ref{lemma:crown-lemma2}. If there is a matching saturating $\mathcal{C}$, the number of colors is bounded by $|V(G)| = n$. Otherwise, we have a crown decomposition $(C,H,R)$ with $C \subseteq \mathcal{C}$ and $H \subseteq V(G)$. Since there is a matching $M$ from $H$ to $C$ saturating $H$ and $N(C) =H$, we can safely color the vertices of $H$ using the colors $C$ according to the matching $M$ and reduce the instance. Thus, the number of colors is reduced to $n$. Further ideas are used to obtain an {FPT} algorithm.

\subsection{Longest Cycle parameterized by vertex cover size}
\label{sec:long-cycle}

\defparproblem{{\sc Longest Cycle (vc)}}{A graph $G$, integers $k$ and $\ell$ and a vertex cover $S$ of $G$ such that $|S| = k$.}{$k$.}{Does $G$ contain a cycle on $\ell$ vertices?}

Bodlaender et al.~\cite{bodlaender2013kernel} showed that {\sc Longest Cycle (vc)} admits a kernel with $\OO(k^2)$ vertices. Here, they construct an auxiliary bipartite graph with pairs of vertices on one side and $I = V(G) \setminus S$ on the other side with edges added if both vertices in a pair are adjacent to a vertex in $I$. We now apply Lemma \ref{lemma:crown-lemma2}.
If there is a matching saturating $I$, the overall number of vertices is $\OO(k^2)$. Otherwise, we have a crown decomposition $(C,H,R)$ with $C \subseteq I$ and $H \subseteq S$. Let $M$ be a matching from $H$ to $C$ saturating $H$. We claim that for any vertex $v \in C$ not saturated by $M$, we can reduce the instance to $(G - v, k, \ell, S)$. Suppose a solution cycle $C'$ in $G$ contains $v$. Since $I$ is an independent set, the two neighbors of $v$ in $C$ are $s_1, s_2 \in S$. Then note that there is another vertex $v' \in I$ adjacent to both $s_1$ and $s_2$ as there is an edge in the matching $M$ corresponding to $v'$ and the pair $(s_1, s_2) \in H$. This gives us a cycle $C''$ by replacing $v$ with $v'$ in $G-v$.

%
%

\subsection{Generic Crown Decomposition and its application to packing and covering problems}
\label{sec:gen-crown}

Chen et al.~\cite{ChenFSWY19} introduced a general notion of crowns and used it to obtain an $\OO(k)$ vertex kernel for {\sc $T_r$-Packing} problem (for each fixed integer $r$) where $T_r$ is a tree with $r$ edges.

It is folklore that many covering problems in graphs can be formulated as an instance of {\sc $d$-Hitting Set} defined as follows.

\defparproblem{{\sc $d$-Hitting Set}}{A family $\FF$ of sets over a universe $U$, where each set in $\FF$ has size at most $d$, and an integer $k$.}{$k$.}{Does there exist a subset $X \subseteq U, |X| \leq k$ such that $X$ contains at least one element from each set in $\FF$?}

Similarly, many packing problems in graphs can be formulated as an instance of {\sc $d$-Set Packing}  defined as follows.

\defparproblem{{\sc $d$-Set Packing}}{A family $\FF$ of sets over a universe $U$, where each set in $\FF$ has size at most $d$, and an integer $k$.}{$k$.}{Does there exist a subset $\FF' \subseteq \FF, |\FF'| =k$ such that every element in $U$ is contained in at most one set in $\FF'$?}

It is folklore that the {\sc $d$-Hitting Set} and {\sc $d$-Set Packing} problems can be formulated as a dominating set problem and their duals in bipartite graphs which are defined as follows.

\defparproblem{{\sc $d$-Red/Blue Dominating Set}}{A bipartite graph $G = (R \uplus B, E)$, where $\forall x \in B: deg(x) = d$, and an integer $k$.}{$k$.}{Does there exist a subset $X \subseteq R, |X| \leq k$ such that $X$ such that $N(X) = B$?}

\defparproblem{{\sc $d$-Red/Blue Distance-$3$ Packing}}{A bipartite graph $G = (R \uplus B, E)$, where $\forall x \in B: deg(x) = d$, and an integer $k$.}{$k$.}{Does there exist a subset $X \subseteq B, |X| \geq k$ such that each $x \in R$ is a neighbor of at most one $x \in X$?}

We formally state the definition of a {\em generic crown decomposition} below.
\begin{definition}[\cite{ChenFSWY19}]
\label{defn:gen-crown-bipartite} 
Let $G = (V = A \uplus B, E)$ be a bipartite graph such that for all $x \in B$, $deg_G(x) = d$.
Then, a {\em generic crown decomposition} $(C, H, X)$ of $G$ is given by a {\em head} $H \subseteq A$ and a {\em crown} $C \subseteq (A \uplus B) \setminus (H \cup C)$ satisfying the following properties
\begin{itemize}
	\item $H$ is a separator in $G$, one collection of connected components being $C$.
	\item $N(H) \supseteq C \uplus B$. 
	\item There is an injective mapping $M: H \rightarrow C \cap B$ such that 
\begin{itemize}
\item  for all $x \in H$, $xM(x) \in E(G)$, and 
\item each $v \in A$ is a neighbor of at most one $x \in M(H)$.
\end{itemize}	
   \item $X = (A \uplus B) \setminus (H \cup C)$ are the remaining vertices.
\end{itemize} 
\end{definition}

The authors proved that $(G,k)$ is a YES instance for {\sc $d$-Red/Blue Dominating Set} (or {\sc $d$-Red/Blue Distance-$3$ Packing} resp.) if and only if $(G - (H \cup C),k)$ is YES instance where $(C, H, X)$ is a {\em generic crown decomposition} of $G$.

The authors devised a framework for obtaining crown rules for covering and packing problems in graphs. They define a notion of crown decomposition specific to the problem. They then reduce the problem to {\sc $d$-Red/Blue Dominating Set} (or {\sc $d$-Red/Blue Distance-$3$ Packing} resp.) where the problem-specific crown decomposition corresponds to a generic crown decomposition. From the safety of the reduction corresponding to the generic crown decomposition, the safety of the problem-specific crown decomposition follows.

The authors demonstrate this framework in the context of the following covering and packing problems.

\defparproblem{{\sc $r$-Edge-Tree-Covering}}{An undirected graph $G$ and an integer $k$.}{$k$}{Is there a set $S$ of at most $k$ vertices such that $G - S$ has no tree with $r$ edges as a subgraph?}

The dual version of this problem is defined as follows.

\defparproblem{{\sc $r$-Edge-Tree-Packing}}{An undirected graph $G$ and an integer $k$.}{$k$}{Is there a set of at least $k$ pairwise disjoint collection of trees with $r$ edges each in $G$?}

They \cite{ChenFSWY19} define {\em $r$-tree crown decomposition} as follows.

\begin{definition}
\label{defn:T-r-crown}[\cite{ChenFSWY19}]
An {\em $r$-tree crown decomposition} of a graph $G = (V, E)$ is a partitioning of $V(G) = C \uplus H \uplus X$ satisfying the following properties.
\begin{itemize}
	\item $H$ (the head) separates $C$ from $X$ in $G$.
	\item $G[C]$ ($C$ is the crown) induces no subgraph that is a tree with $r$ edges.
	\item There are $r$ injective mappings $\pi_1,\ldots,\pi_r$ from $H$ to $C$ called {\em witness functions} such that for all $i \neq j$, $\pi_i(H) \cap \pi_j(H) = \emptyset$ and for each $v \in H$, the vertex set $\{p\} \cup \{\pi_i(v) \mid 1 \leq i \leq r\}$ forms a tree with $r$ edges in $G$.
\end{itemize}
\end{definition}

Given an instance $(G,k)$ of  {\sc $r$-Edge-Tree-Packing}, the authors construct a bipartite graph $G' = (R \uplus B, E')$ as follows. Let $R=V$ and $B$ be the set of all trees with $r$ edges present in $G$. Add an edge $uv \in E', u \in R, v \in B$ if and only if the vertex $u$ is part of the tree corresponding to $v$. 

Let us focus on an $r$-tree crown decomposition $(C,H,X)$ in $G$. The witness functions $\pi_1,\ldots,\pi_r$ from $H$ to $C$ correspond to a mapping $M$ from $H$ to some of the trees in $B$; in particular $v \in H$ is mapped to the tree $\{v, \pi_1(v),\ldots,\pi_r(v)\}$. From the definition of $\pi_i, i \in [r]$, one can conclude that $M$ is a matching and $M(H)$ is a packing. Thus $(C',H,X')$ is a generic crown decomposition in $G'$ where $C' = C \uplus M(H)$ and $X' = V(G') \setminus (H \cup C')$. Since $(G',k)$ is a YES-instance of {\sc $r$-Red/Blue Distance-$3$ Packing} if and only if $(G' - (H \cup C'),k-|H|)$ is a YES-instance of {\sc $r$-Red/Blue Distance-$3$ Packing}, one can conclude that $(G,k)$ is a YES-instance of {\sc $r$-Edge-Tree-Packing} if and only if $(G - (H \cup C),k-|H|)$ is a YES-instance of {\sc $r$-Edge-Tree-Packing}.

Using the above reduction, the authors provide a kernel with at most $(r^2 + 1)(r+1)k$ vertices for {\sc $r$-Edge-Tree-Packing}. The authors also devise a methodology to transfer a kernelization result from a packing problem to the corresponding covering problem with a similar size (See \cite{ChenFSWY19} for more details). Thus a kernel with at most $(r^2 + 1)(r+1)k$ vertices for {\sc $r$-Edge-Tree-Covering} also follows.


\subsection{Other examples}
\label{sec:crown-other-examples}

More examples of the application of crown decompositions and some variants include 
\begin{itemize}
\item a kernel with $\OO(k^3)$ vertices for {\sc Triangle Packing} (packing vertex disjoint triangles as subgraphs)~\cite{fellows2004finding},
\item an kernel with $\OO(k^2)$ vertices for {\sc $k$-Internal Spanning Tree} (finding a spanning tree with at least $k$ internal vertices)~\cite{prieto2005reducing},
\item a linear-sized kernel for several variants of {\VC}~\cite{chlebik2008crown},
\item a $(2d -1)k^{d-1} + k$ sized kernel for {\sc $d$-Hitting Set}~\cite{abu2010kernelization},
\item an $\OO(k^{d-1})$ sized kernel for {\sc $d$-Set Packing}~\cite{abu2010improved}, and
\item  a kernel with $7k$ vertices for {\sc $P_3$-Packing} (packing vertex disjoint $P_3$'s as subgraphs)~\cite{wang2010improved}. Later, Xiao and Khu \cite{XiaoK17} improved this result into a kernel with $5k$ vertices using crown decomposition with some critical observations.
\end{itemize} 

\section{(Basic) Expansion Lemma and Applications}
\label{sec:expansion-lemma-basic}

In this section, we provide the definition of `$q$-expansion' and expansion lemma that generalizes hall's theorem (Proposition \ref{prop:hall-theorem}).
In the following subsections, we give a summary of how expansion lemma has been used for several problems, e.g. {\sc $p$-Component Order Connectivity}, {\FVS}, {\CVD} and several other graph theoretic problems.

\begin{definition}[$q$-Expansion]
\label{defn:q-expansion}
Let $G = (V = A \uplus B, E)$ be a bipartite graph. 
For a positive integer $q$, a set of edges $M \subseteq E(G)$ is called a {\em $q$-expansion of $A$ into $B$} if
\begin{itemize}
	\item every vertex of $A$ is incident to exactly $q$ edges of $M$, and
	\item $M$ saturates exactly $q|A|$ vertices in $B$.
\end{itemize}
\end{definition}

For $q=1$, a $q$-expansion is just a matching. We now have expansion lemma, which was introduced by Thomass\'e~\cite{Thomasse2010}.

\begin{lemma}[{\qExpLem}~\cite{Thomasse2010}]
\label{lem:expansion-lemma}
Let $q \geq 1$ be a positive integer and $G = (V = A \uplus B, E)$ be a bipartite graph 
such that
\begin{itemize}
	\item $|B| \geq q|A|$, and
	\item there are no isolated vertices in $B$.
\end{itemize}
Then, there exists nonempty vertex sets $X \subseteq A$ and $Y \subseteq B$ such that
\begin{enumerate}
	\item\label{prop:expansion-property-1} there is a $q$-expansion of $X$ into $Y$, and
	\item\label{prop:expansion-property-2} $N_G(Y) \subseteq X$. 
\end{enumerate}
Furthermore, the sets $X$ and $Y$ can be found in polynomial time in the size of $G$ (see Figure \ref{fig:expansion-lemma-basic} for an illustration for $q = 2$).
\end{lemma}

\begin{figure}[t]
\centering
	\includegraphics[scale=0.3]{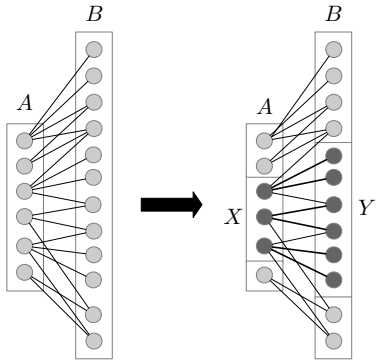}
	\caption{An illustration of Lemma \ref{lem:expansion-lemma} with $q = 2$.}
\label{fig:expansion-lemma-basic}
\end{figure}

When $q=1$, note that the sets $Y, X$ and $V(G) \setminus (X \cup Y)$ form a crown decomposition with crown, head, and body, respectively. 
Hence, expansion lemma can be seen as a generalization of Lemma~\ref{lemma:crown-lemma2}. 
In particular,  it allows us to work
with objects that need to be hit that can be of size more than $2$.
It gives a way to ``encode'' disjoint objects of size $q+1$ such that any solution must hit these objects. 
We use this to argue that there is a solution containing $X$, noting that every object $Y$ is part of can be transformed into a solution containing $X$.

We illustrate this with a few examples below. To encode objects in this way, we usually construct an auxiliary bipartite graph where 
we apply expansion lemma.

\subsection{$p$-Component Order Connectivity}
\label{sec:p-comp-oc}

In {\sc $p$-Component Order Connectivity}, the goal is to find a subset
of
$k$ vertices in a given graph, such that every component of the remaining graph has size at
most
$p$, for a fixed integer $p$. More formally,

\defparproblem{{\sc $p$-Component Order Connectivity}}{A graph $G$ and an integer $k$.}{$k$.}{Does $G$ contain a subset $S$ of vertices, $|S| \leq k$, such that every connected component of $G - S$ is of size at most $p$?}

Note that when $p=1$, $G-S$ is an independent set.
Thus {\sc $p$-Component Order Connectivity} is a generalization of {\sc Vertex Cover}.
We explain how a kernel with $\OO(p^3 k)$ vertices can be obtained for this problem.
Note that improved kernels are available for this problem using stronger versions of expansion lemma (see Section \ref{sec:new-double-expansion-lemma}), we give an $\OO(k^3 p)$ kernel to illustrate a simple application of Lemma \ref{lem:expansion-lemma}.
The results in the literature are improvements of this result using some new tool (available in Section \ref{sec:new-double-expansion-lemma}).

We can get an approximate solution $X$ of size at most $(p+1)k$ 
by repeatedly finding connected
subsets of size $p+1$, adding it to $X$ and deleting it from $G$ (if $|X| > (p+1)k$, then it is a no-instance). We can safely remove all the components in $G-X$ of size at most $p$ as its vertices cannot be part of any optimal solution. We now use an auxiliary bipartite graph $Q$ with $X$ on one
side
and components of $G-X$ on the other side. We add an edge between a vertex $u \in X$ and
a component $D$ in $G-X$ if $u$ has a neighbor in $D$. If the number of components in $G-X$ is more
than $p(p+1)k$, then the conditions for Lemma \ref{lem:expansion-lemma} are satisfied with $q=p$. Thus, we get a head $H \subseteq X$ and
a crown
$C$ that is a subset of connected components in $G-X$, such that there is a $p$-expansion of $H$ into $C$ and $N_Q(C) \subseteq H$.
Our reduction rule, like that of {\sc Vertex Cover}, is to include all of $H$ into the solution, remove $H$ and components of $C$ and drop $k$ by $|H|$.
We again aim to claim that for an optimal solution $S$, the set $S' = (S \setminus (H \cup C)) \cup H$ is also an optimal solution.
The $p$-expansion gives a collection of $|H|$ disjoint connected sets of size
$(p+1)$ that any solution needs to hit. Hence $|S'| \leq |S|$.
We use that $N_Q(C) \subseteq H$ and the graph on $C$ contains only components
of
size at most $p$ to claim that $S'$ is indeed a solution. This is because all the connected sets of size more than $p$ that vertices in components of $C$ are part of, must contain a vertex in $H$. These are hit by $S'$.
Thus the number of connected components in $G-X$ is bounded by $p(p+1)k$.
Each
of these components has size at most $p$. Hence we get a kernel with $\OO(p^3k)$ vertices.

\subsection{Undirected Feedback Vertex Set}

 
\defparproblem{{\FVS} (FVS)}{An undirected graph $G = (V, E)$ and an integer $k$.}{$k$.}{Is there a set $S \subseteq V(G)$ of at most $k$ vertices such that $G - S$ is acyclic?}

Thomass\'e \cite{Thomasse2010} gave a kernel with $\OO(k^2)$ vertices for {\FVS}. The main idea is to use expansion lemma to bound the maximum degree of the graph to $\OO(k)$. 

Suppose a vertex $v$ in the graph has degree more than $11k$. In polynomial time, one can determine~\cite{gallai1961maximum,FLSZ19} that 
$v$ is part of at least $k+1$ cycles with $v$ as the only intersection (in which case $v$ is forced 
into the solution) or find a set $S_v$ of size at most $3k$ that hits all the cycles passing through $v$. 
The next step is to bound the number of components in the graph $G - (\{v\} \cup S_v)$ that has a vertex adjacent to $v$.
Firstly, note that $v$ is adjacent to exactly one vertex in such components as otherwise, there is a cycle containing $v$ not intersecting $S_v$.
Such components can be disregarded in this graph which can contain cycles.
Note that the number of such components is at most $k$ as they form disjoint cycles. 

Consider the auxiliary bipartite graph where one partition is $S_v$ and the other partition is the remaining components of  $G - (\{v\} \cup S_v)$ and add edges if a vertex in $S_v$ is incident to a vertex in the component.
If the number of components is more than $6k$, then Lemma \ref{lem:expansion-lemma} can be used with $q=2$ to find subsets with each vertex in $S_v$ forming a $2$-expansion.
Note that since $v$ is adjacent to these components, this provides us a set of cycles containing $v$ corresponding to each $2$-expansion.
Using the presence of these cycles, it can be proven that any cycles hit by vertices in these components can be hit by either $v$ or all the vertices in the expansion subset of $S_v$.
The authors encode this by creating disjoint $2$-cycles containing $v$ and each vertex in the expansion subset of $S_v$ if not already present.

This bounds the number of such connected components to $2|S_v| \leq 6k$. 
Thus, the number of vertices adjacent to $v$ is bounded by $6k + |S_v| + k \leq 10k$. At most, $k$ of these vertices can have $2$ edges to $v$ as we can delete $v$ if it is part of $k+1$ cycles intersecting only at $v$.
Thus, the degree of $v$ is bounded by $11k$.

From other reduction rules, it is possible to prove that the graph has minimum degree three.
Finally, in this case, it can be proved that {\FVS} admits a kernel with $\OO(k\Delta)$ vertices where $\Delta$ is the maximum degree of $G$. Thus, we have a kernel with $\OO(k^2)$ vertices as $\Delta \leq 11k$.

Recently, Iwata \cite{Iwata17} has also obtained a kernel with $\OO(k^2)$ vertices in $\OO(m + n)$ time.

\subsection{Cluster Vertex Deletion}
\label{subsection:cvd-subq}
\defparproblem{{\sc Cluster Vertex Deletion (cvd)}}{A graph $G$ and an integer $k$.}{$k$.}{Does $G$ contain a subset $S$ of at most $k$ vertices such that every connected component of $G\setminus S$ is a clique?}


The {\sc Cluster Vertex Deletion} problem can be formulated as a {\sc $3$-Hitting Set} problem with $U= V(G)$ and  $\FF$ being the vertex sets of all induced paths on three vertices ($P_3$'s) of $G$. Abu-Khzam~\cite{abu2010kernelization} obtained a kernel with $\OO(k^2)$ elements for {\sc $3$-Hitting Set}. This kernel can be adapted to obtain kernels with $\OO(k^2)$ vertices for {\sc Cluster Vertex Deletion}.

Fomin et al. \cite{FominLLSTZ19} have obtained a kernel with $\OO(k^{5/3})$ vertices using Lemma \ref{lem:expansion-lemma} and more sophisticated techniques. We will see a brief overview of this kernel here.

\medskip \noindent{\bf Bounding the number of cliques.} We give a simple argument to show how Lemma \ref{lem:expansion-lemma} can be used to bound
the number of cliques to $6k$.

First, find a maximal collection $\mathcal{P}$ of vertex disjoint $P_3$'s in $G$. If $|\mathcal{P}| > k$, we return a no-instance as there is no optimal cluster vertex deletion set of size at most $k$.
Otherwise, set $S$ as the subset of vertices in $\mathcal{P}$. Note that $|S| \leq 3k$. We use an auxiliary bipartite graph $Q$ with $S$ on one side and components (cliques) of $G-S$ on the other side. 
Add an edge between a vertex $u \in S$ and a component $D$ in $G-S$ if $u$ has a neighbor in $D$.
	Remove isolated cliques if any. If the number of remaining cliques is more than $6k$, then the conditions are satisfied for Lemma \ref{lem:expansion-lemma} with  $q=2$.
	This provides a head 
$H \subseteq S$ and a crown $C$ that is a subset of components in $G-S$, such that there is a $2$-expansion from $H$ to $C$ and $N_Q(C) \subseteq H$.
The reduction rule like that of {\sc Vertex Cover} is to include all of $H$ into the solution, remove $H$ and the components of $C$ and drop $k$ by $|H|$. This is because it is easy to show that for any optimal solution $X$, the set $X' = X \setminus (H \cup C) \cup H$ is also an optimal solution.
The $2$-expansion gives a collection of $|H|$ disjoint $P_3$'s that 
any solution needs to hit. Hence $|X'| \leq |X|$.
Then they use that $N_Q(C) \subseteq H$ and that the graph on $C$ is a cluster graph to claim that $X'$ is indeed a solution. This is because all the $P_3$'s 
that vertices in components of $C$ are part of, have to contain a vertex in $H$. These are hit by $X'$.
Thus the number of cliques in $G-S$ is bounded by $6k$.

\medskip \noindent{\bf A marking procedure and a weak form of crown decomposition.} 
After bounding the number of cliques in $G - S$ by $6k$, the authors \cite{FominLLSTZ19} perform a marking procedure where we repeatedly associate an edge in $G-S$ with each vertex in $S$ such that exactly one endpoint of the edge is adjacent to that vertex.
If $k+1$ edges are associated with a vertex in $S$, it can be concluded that this vertex must be in the solution as it is the intersection of $k+1$ $P_3$'s and reduce the instance accordingly.
Otherwise, if roughly $k^{2/3}$ vertices of $S$ have roughly $k^{2/3}$ edges associated with them, call the run `successful' and mark and delete these vertices.
This marking procedure is repeated until a successful run is obtained.

Let $U$ be the set of all marked vertices and $M$ be the set of endpoints of the edges marked in the unsuccessful final run. Let $L = S \setminus U$. The sets $U, V(G) \setminus S$ and $L$ can be viewed as a weak form of crown decomposition with $U$ being the head, $V-S$ being the crown, and $L$ being the body. When crown decomposition is applied to obtain kernels, usually it can be concluded that there is a cluster vertex deletion set that contains all of the head, none of the crown, and all the edges of the crown are within the head. Similarly, in the weak crown decomposition, it can be concluded that {\em most} of the vertices of the head $U$ go into a solution, {\em almost none} of the vertices of the crown $V(G) \setminus S$ is in the solution and {\em almost all} of the edges from $V(G) \setminus S$ are within $U$.

\medskip \noindent{\bf Identifying good/bad cliques and bounding remaining vertices.}  Nonetheless, there is no straightforward reduction rule for this weak crown decomposition that the authors obtain.
Towards this, they classify the cliques in $G-S$ into good cliques and bad cliques (based on the weak form of crown decomposition). Based on its properties, the number of vertices of bad cliques can be bounded to $\OO(k^{5/3})$. For a vertex $s \in S$ and a good clique $C$, the authors \cite{FominLLSTZ19} identify a small subset of the vertices of $C$ based on the number of neighbors of $C$ in $s$.
Then they show that good cliques exhibit a vertex-cover-like behavior; any solution either contains $s$ or the small side of a good clique.
This property can be exploited by using Lemma \ref{lem:expansion-lemma} to bound the number of vertices of the good cliques to $\OO(k^{5/3})$ as well.


\subsection{Other Deletion Problems}
\label{sec:other-expansion-lemma-applications}

In this section, we list some other deletion problems where expansion lemma was used to obtain kernels.

\defparproblem{{\dPVC} ($d$-PVC)}{An undirected graph $G = (V, E)$ and an integer $k$.}{$k$}{Is there a set $S \subseteq V(G)$ of at most $k$ vertices such that $G - S$ has no path of length $d$?}

Here, $d$ is a fixed integer. There have been several previous works on $d$-PVC with $d = 4, d = 5$, etc.
Cereveny and Suchy~\cite{CervenyS19} provided parameterized algorithms for $4$-PVC and $5$-PVC.
Later, Cerveny et al.~\cite{CervenyCS21} obtained a polynomial kernel for {\dPVC} with $\OO(k^2)$ vertices for $d=4$ and $d=5$ using Lemma \ref{lem:expansion-lemma}. They also obtained a polynomial kernel for {\dPVC} with $k^4 d^{\OO(d)}$ vertices and edges for every fixed $d$.

\defparproblem{{\sc VC-2-Mod}}{A graph $G = (V, E)$, a set $S \subseteq V(G)$ such that every vertex of $G - S$ has degree at most 2, and an integer $k$.}{$|S|$}{Does $G$ have a vertex cover of size at most $k$?}

Let $|S| = \ell$. Majumdar et al.~\cite{MRS18} used Lemma \ref{lem:expansion-lemma} to obtain a kernel with $\OO(\ell^5)$ vertices for {\sc VC-2-Mod}.
Note that the connected components of the graph $G - S$ include isolated vertices, paths, and cycles. Lemma \ref{lem:expansion-lemma} is used to bound the number of odd cycle components to $\OO(\ell^3)$. An auxiliary bipartite graph is created with sets of size three in $S$ that form an independent set on one side, and the odd cycle components on the other side, adding edges based on a notion called blocking set. 
After applying Lemma \ref{lem:expansion-lemma} with $q=4$, each $3$-sized independent set in the expansion has four private odd cycles. It can be proved that it is safe to delete one of the four odd cycle components resulting in an instance with $O(\ell^3)$ many odd cycles (without increasing the size of $S$).
Finally, a kernel with $\OO(\ell^5)$ vertices can be obtained with additional arguments.

\defparproblem{{\PWOneVD}}{An undirected graph $G = (V, E)$ and an integer $k$.}{$k$.}{Is there a set $S \subseteq V(G)$ of at most $k$ vertices such that $G - S$ has pathwidth at most one?}

Philip et al.~\cite{PhilipRV10} provided a kernel with $\OO(k^2)$ vertices for {\PWOneVD}.
Cygan et al.~\cite{CyganPPW12} used Lemma \ref{lem:expansion-lemma} to give a simpler kernel with $\OO(k^2)$ vertices.

\defparproblem{{\BlockVD}}{An undirected graph $G = (V, E)$ and an integer $k$.}{$k$.}{Is there a set $S \subseteq V(G)$ of at most $k$ vertcies such that $G - S$ is a block graph?}

Kim and Kwon~\cite{KinKwon15} initiated the study of {\BlockVD} problem from parameterized complexity perspective and provided a kernel with $\OO(k^9)$ vertices.
Later, Agrawal et al. \cite{AgrawalKLS16} improved it by giving a kernel with $\OO(k^4)$ vertices using Lemma \ref{lem:expansion-lemma}.

\defparproblem{{\ChordVD}}{An undirected graph $G = (V, E)$ and an integer $k$.}{$k$.}{Is there a set $S \subseteq V(G)$ of at most $k$ vertices such that $G - S$ has no induced cycle of length at least four?}

Jansen et al.~\cite{JansenP18} provided a kernel with $\OO(k^{161}\log ^{58} k)$ edges for {\ChordVD}.
Later, Agrawal et al.~\cite{agrawal2018feedback} improved their result using Lemma \ref{lem:expansion-lemma} into a kernel with $\OO(k^{12}\log^2 k)$ edges. In addition, Agrawal et al.~\cite{agrawal2018feedback} used the notion of {\em independence degree} of a vertex to obtain this result.

\defparproblem{{\ECT}}{An undirected graph $G = (V, E)$ and an integer $k$.}{$k$.}{Is there a set $S \subseteq V(G)$ of at most $k$ vertices such that $G - S$ has no even cycle?}

Misra et al.~\cite{MisraRRS12} initiated the study of {\ECT} problem from a parameterized complexity perspective. They proved the problem to be FPT and gave a kernel with $\OO(k^2)$ vertices. 
Their kernelization algorithm also uses Lemma \ref{lem:expansion-lemma} as a crucial tool to obtain such a kernel upper bound.

\defparproblem{{\OFVDS}}{A directed graph $D = (V, A)$ and a positive integer $k$.}{$k$.}{Is there a set $S \subseteq V(D)$ of at most $k$ vertices such that $D - S$ is an out-forest?}

 Minch and van Leeuwen \cite{mnich2017polynomial} provided a kernel with $\OO(k^3)$ vertices for {\OFVDS}. 
Later, Agrawal et al.~\cite{AgrawalSSZ18} used Lemma \ref{lem:expansion-lemma} to get an improved kernel with $\OO(k^2)$ vertices.

\defparproblem{{\IFVS}}{An undirected graph $G = (V, E)$ and an integer $k$.}{$k$.}{Is there a set $S \subseteq V(G)$ of at most $k$ vertices such that $S$ is an independent set and $G - S$ is acyclic?}

Misra et al. \cite{MPRS12} studied {\IFVS} problem from a parameterized complexity perspective. 
This problem is a variation of {\FVS} with additional constraints imposed for the feedback vertex set.
They use Lemma \ref{lem:expansion-lemma} to obtain a kernel with $\OO(k^3)$ vertices and edges.

A common theme for the problems 
{\PWOneVD}, {\BlockVD}, {\ChordVD}, {\ECT}, {\OFVDS} and {\IFVS} is that they all have some infinite set of cycles as forbidden subgraphs. The application of Lemma \ref{lem:expansion-lemma} (expansion lemma) for all these problems is inspired by its application by Thomass\'e \cite{Thomasse2010} for {\FVS}. The $q$-expansion from the lemma is used to identify sets of cycles that pairwise intersect in one or two vertices, which has to be hit by the solution. This is used to bound a degree measure for vertices $v$ of $G$ (the measure depends on the problem) by bounding the number of components associated with $v$.

\subsection{Graph Packing Problems}
\label{sec:packing-problems}

Expansion lemma is also used on packing problems. We list a few below.

\noindent
{\sc Cycle Packing} is a well-studied problem in parameterized complexity and kernelization. 
The {\sc Disjoint Cycle Packing} problem asks if there are at least $k$ cycles in a graph that are pairwise disjoint.
It is well-known that {\sc Disjoint Cycle Packing} is FPT and admits no polynomial kernelization unless {\nka} \cite{BodlaenderTY11}.
After that, Agrawal et al. \cite{ALMMS18} introduced relaxation criteria in disjointness constraints and defined the following two problems {\tPDCP} and {\tADCP} for a fixed integer $t$.

\defparproblem{{\tPDCP} ($t$-PDCP)}{An undirected graph $G$ and an integer $k$.}{$k$.}{Are there at least $k$ distinct cycles $C_1,\ldots, C_k$ in $G$ such that for every $i \neq j$, $|V(C_i) \cap V(C_j)| \leq t$?}

\defparproblem{\tADCP}{An undirected graph $G$ and an ingeger $k$.}{$k$.}{Are there at least $k$ distinct cycles $C_1,\ldots, C_k$ in $G$ such that every vertex of $G$ appears in at most $t$ of these cycles?}

Agrawal et al. \cite{ALMMS18} proved that when $t = 1$, {\tPDCP} admits a kernel with $\OO(k^4\log k)$ vertices. This makes crucial use of Lemma \ref{lem:expansion-lemma}.
If $t = |V(G)|$, then {\tPDCP} is polynomial-time solvable.

\defparproblem{\ThreePathPacking}{An undirected graph $G = (V, E)$ and an integer $k$.}{$k$}{Are there $k$ pairwise vertex disjoint induced $P_3$'s in $G$?}

A simple application of Lemma \ref{lem:expansion-lemma} can obtain a kernelization with $\OO(k^2)$ vertices.
Fomin et al. \cite{FominLLSTZ19} provided an improved kernel with $\OO(k^{5/3})$ vertices for {\ThreePathPacking}.
They have used Lemma \ref{lem:expansion-lemma} 
to obtain this result.
Later, Bessy et al. \cite{BessyBTW23SODA} introduced a new technique called {\em rainbow matching} and have improved this result into a kernel with $\OO(k)$ vertices.

\subsection{Other Problems}
\label{sec:other-problems}

In this section, we discuss some applications of expansion lemma other than hitting or packing problems in graphs.

\defparproblem{{\sc Set Splitting}}{A universe $\UU$, a family $\FF$ of subsets of $\UU$, and a positive integer $k$.}{$k$.}{Is there a bipartition of $\UU$ such that at least $k$ sets have non-empty intersection with both parts?}

Lokshtanov et al. \cite{lokshtanov2009even} gave a kernel for  {\sc Set Splitting} parameterized by $k$ with at most $2k$ sets and $k$ elements. They use a notion of strong cut-sets in hypergraphs which is closely related to crowns in graphs.

\defparproblem{{\sc Maximum Internal Spanning Tree}}{A graph $G$ and a positive integer $k$.}{$k$.}{Is there a spanning tree of $G$ with at least $k$ internal vertices?}

Fomin et al. \cite{fomin2013linear} use Lemma \ref{lem:expansion-lemma} with $q=2$ to give a $3k$-vertex kernel for {\sc Maximum Internal Spanning Tree}.

\section{Recent Developments}
\label{sec:new-double-expansion-lemma}

In this section, we discuss recent developments that include several variations of expansion lemma along with their applications.
Such new variations have appeared recently. 
Some of those variants are generalizations, and some of them have stronger and special implications.
This section is primarily devoted to the notions of stronger expansion lemma, double expansion lemma, balanced expansion lemma, balanced crown decomposition, and additive expansion lemma (see the relevant subsections for details).

\subsection{Weighted $q$-Expansion Lemma}
\label{sec:weighted-q-expansion-mithilesh}

We introduced {\sc $p$-Component Order Connectivity} in Section \ref{sec:p-comp-oc} and explained how a kernel with $\OO(p^3 k)$ vertices can be obtained.
In this section, we explain how to obtain a 
kernel with $\OO(p^2 k)$ vertices for the same problem using a weighted version of expansion lemma. The weighted version of expansion lemma was introduced by Kumar and Lokshtanov \cite{KL16}, who gave a kernel with $2pk$ vertices of this problem. A variant of this lemma was also proposed by Xiao \cite{Xiao17a}, who also gave an $\OO(pk)$ kernel for this problem.

We use a simpler version of the weighted expansion lemma from \cite{FLSZ19}. We start with some definitions.

\begin{definition}[Weighted $q$-Expansion]
Consider a bipartite graph $G = (V = A \uplus B, E)$ with a weight function $w: V(G) \rightarrow \{1, \ldots, W\}$ for an integer $W$. Given a subset $X \subseteq V(G)$, let $w(X) = \sum\limits_{x \in X}^{} w(x)$. For an integer $q \geq 1$, a function $f :B \rightarrow A$ is called a {\em weighted $q$-expansion of $A$ into $B$}  if
\begin{itemize}
\item for every $b \in B$, $f(b) \in N(b)$ , and
\item for every $a \in A$, $w(f^{-1}(a)) \geq q - W +1$.
\end{itemize}
\end{definition}

When $W = 1$, a weighted $q$-expansion corresponds to a $q$-expansion as for $a \in A$, $w(f^{-1}(a)) \geq q$ and $w(f^{-1}(a))$ becomes the number of neighbors of $a$ in $B$.

\begin{lemma}[Weighted Expansion Lemma \cite{FLSZ19}]
\label{lemma:weighted-expansion-lemma}
Let $q$ and $W$ be positive integers. Consider a bipartite graph $G = (V = A \uplus B, E)$ with weight function $w: V(G) \rightarrow \{1, \ldots, W\}$ such that 
\begin{enumerate}
\item $w(B) \geq q|A|$, and 
\item there are no isolated vertices in $B$.
\end{enumerate}

Then there exist non-empty vertex sets $X \subseteq A$, $Y \subseteq B$ and a function $f : Y \rightarrow X$ such that
\begin{itemize}
\item $f$ is a weighted $q$-expansion of $X$ into $Y$, and
\item no vertex in $Y$ has a neighbor outside $X$, that is, $N(Y) \subseteq X$.
\end{itemize}
Furthermore, the sets $X$ and $Y$ and the function $f$ can be found in time $\OO(mn^{1.5}W^{2.5})$.
\end{lemma}


We first provide a kernel with $\OO(p^2k)$ vertices for {\sc $p$-Component Order Connectivity} for the problem (discussed in \cite{FLSZ19}) as a precursor to the kernel with $2pk$ vertices by Kumar and Lokshtanov~\cite{KL16}.
Recall the auxiliary bipartite graph in Section \ref{sec:p-comp-oc} with modulator $X$ on one side and 
$\mathcal{C}$, the connected components of $G-X$, in the other one. For each connected component in $\mathcal{C}$, we assign its weight as the number of vertices in the component. Note that these weights are at most $p$. For all the vertices in $X$, we assign weight $1$.

 We now apply Lemma \ref{lemma:weighted-expansion-lemma} with $q = 2p-1$ to get a subset $H \subseteq X$ and a subset $C \subseteq \mathcal{C}$ such that there is a weighted $q$-expansion of $H$ into $C$. For any optimal solution $S$, we look at the set $S' = S \setminus (H \cup C) \cup H$. The weighted $q$-expansion gives us disjoint connected sets in the auxiliary bipartite graph, each of whose weight is at least $q - p+ 1 \geq p$. Thus, we can conclude that the weighted $q$-expansion represents a connected set of size at least $p+1$ in $G$, which any solution needs to hit. Similar to the argument in Section \ref{sec:p-comp-oc}, we can prove that $S'$ is an optimal solution. Thus we have a reduction rule as long as Lemma \ref{lemma:weighted-expansion-lemma} is applicable. If this is not the case, we have $w(\mathcal{C}) \leq q|X| \leq (2p-1)(p+1)k$ in the auxiliary bipartite graph. From the conclusion that the number of vertices in  $G-X$ is $w(\mathcal{C})$, we have a $\OO(p^2k)$ kernel.

Note that in the kernel with $\OO(p^3 k)$ vertices for {\sc $p$-Component Order Connectivity} described in Section \ref{sec:p-comp-oc}, we were not able to capture the sizes of the components in the auxiliary bipartite graph. We manage to do so using weighted expansions.

Kumar and Lokshtanov~\cite{KL16} use some further ideas and linear programming to improve the kernel size to $2pk$.

\subsection{Stronger Expansion Lemma and Double Expansion Lemma}
\label{sec:new-expansion-double-expansion-lemma}

Fomin et al. \cite{FominLLSTZ19} defined the following stronger notion of a $q$-expansion.
Let $G = (V = A \uplus B, E)$ be a bipartite graph. Let $q > 0$ be an integer, $\hat{A} \subseteq A$, and $\hat{B} \subseteq B$. We say that there is a {\em stronger q-expansion} of $\hat{A}$ into $\hat{B}$ if for all subsets $X \subseteq \hat{A}$, we have $|N_G(X) \cap \hat{B}|  \geq q |X|$. Using this notion, they prove the following version of expansion lemma, which does not need the conditions that $|B| \geq q|A|$ or that there are no isolated vertices in $B$.

\begin{lemma}[Stronger Expansion Lemma~\cite{FominLLSTZ19}]
\label{lemma:new-expansion-lemma}
Let $q$ be a positive integer, and $G = (V = A \uplus B, E)$ be a bipartite graph.
Then there exist (possibly empty) sets $\hat A \subseteq A$ and $\hat B \subseteq B$ such that 
\begin{enumerate}
	\item\label{prop1:new-expansion} $\hat A$ has a stronger $q$-expansion into $\hat B$,
	\item\label{prop2:new-expansion} $N_{G}(\hat B) \subseteq \hat A$, and
	\item\label{prop3:new-expansion} $|B \setminus \hat B| \leq q|A \setminus \hat A|$. 
\end{enumerate}
Moreover, the sets $\hat A$ and $\hat B$ can be computed in polynomial time. See Figure \ref{fig:stronger-expansion-lemma} for an illustration with $q = 2$.
\end{lemma}

\begin{figure}[t]
\centering
	\includegraphics[scale=0.3]{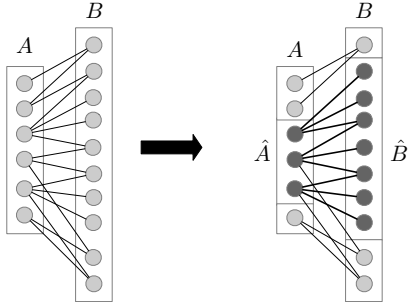}
	\caption{An illustration of Lemma \ref{lemma:new-expansion-lemma} with $q = 2$.}
\label{fig:stronger-expansion-lemma}	
\end{figure}

This lemma is different from Lemma \ref{lem:expansion-lemma} as it does not require that $|B| \geq q|A|$ and there are no isolated vertices in $B$. Also, note that this lemma does not demand that $\hat A$ and $\hat B$ must be non-empty.
If $\hat B = \emptyset$ for some graph $G$, then $\hat A = \emptyset$ as there is a stronger $q$-expansion from $\hat A$ to $\hat B$.
Moreover, since $|B \setminus \hat B| \leq q |A \setminus \hat A|$, we have $|B| \leq q|A|$. 
Therefore, if $|B| > q|A|$, then $\hat B \neq \emptyset$. 

Note that in Lemma \ref{lem:expansion-lemma} (expansion lemma), $|Y| = q|X|$. But in stronger expansion lemma, it could be that $|\hat B| > q|\hat{A}|$. The intuition in applying stronger expansion lemma for a problem instance is that some vertices of $\hat B$ could be safely eliminated from the input graph. Note that since $|B \setminus \hat B| \leq q |A \setminus \hat A|$, to bound the size of $B$ as a function of $|A|$, we might need to remove some of the vertices of  $\hat B$. Observe that for every $u \in \hat A$, there is at least $q$ edges incident. Thus, by keeping exactly $q|\hat A|$ vertices in $\hat B$, the property of a $q$-expansion from $\hat{A}$ is maintained. Thus, the other vertices from $\hat B$, that only has neighbors to $\hat{A}$, could be deleted. 

The following problem {\ArcFourCP} provides an intuition to use the above lemma.


\defparproblem{{\ArcRCP}}{A directed graph $D = (V, A)$ and an integer $k$.}{$k$}{Are there at least $k$ pairwise arc disjoint directed cycles of length $r$ in $D$?}

Babu et al.~\cite{BabuKR22} first proved that for every fixed $r \geq 3$, {\ArcRCP} is NP-Complete in directed graphs of girth $r$, and this NP-Completeness result holds for every $r \geq 4$ when the input graph is bipartite.
After that they provide a kernel with $\OO(k^2)$ vertices and $\OO(k^3)$ arcs for $r = 4$, i.e. for the {\ArcFourCP} problem using stronger expansion lemma (Lemma \ref{lemma:new-expansion-lemma}).
A brief summary of the kernelization algorithm by Babu et al.~\cite{BabuKR22} works as follows.

\begin{enumerate}
	\item Initially, the authors apply some preprocessing rules to the input graph $D$ and ensure that 
		every arc $(u,v)$ of $D$ is contained in at most $(4k-4)$ distinct $4$-cycles, the pairwise intersection of which is $\{u,v\}$.
		\item Let $\XX$ be a collection of maximal set of $4$-cycles such that for every pair of $4$-cycles in $\XX$, there is at most one common vertex. 
		The authors prove that if $|\XX| > 16k^2$, then the input is a YES-instance.
	\item So, it can be assumed that $|\XX| \leq 16k^2$.
	 Let $\PP$ be the set of all $3$-paths that are contained in some 4-cycle in $\XX$. As $|\XX| \leq 16k^2$, $|\PP|$ is $\OO(k^2)$. What is left is to bound the number of vertices in $V(D) \setminus V(\PP)$. 
	\item The authors construct an auxiliary bipartite graph $H$ with bipartitions $A$ and $B$ where $A = \PP$ and $B = V(D) \setminus V(\PP)$. For $u \in B$ and $\{x,y,z\} \in \PP$, there is an edge in $H$ if and only if $\{x, y, z, u\}$ is a 4-cycle in $D$.
	\item One can now apply Lemma \ref{lemma:new-expansion-lemma} with $q = 1$ and get a stronger $1$-expansion with sets $\hat A \subseteq A$ and $\hat B \subseteq B$. It ensures that $|B \setminus \hat B| \leq |A \setminus \hat A|$. Let $\hat M \subseteq E(\hat A, \hat B)$ be the set of edges of $H$ such that every vertex of $\hat A$ is incident to exactly one vertex of $\hat B$ and every vertex of $\hat B$ is incident to exactly one vertex of $\hat A$. 
	\item Finally, the authors apply a reduction rule that deletes the set of vertices from $\hat B$ that is not saturated by $\hat M$. The $4$-cycles that these vertices are part of can be covered by picking one of the vertices from the corresponding $3$-paths in $\hat{A}$. The matching $M$ ensures that there exists a $4$-cycle for each  $3$-path in $\hat{A}$. Thus, the reduction rule is safe.
	\item The reduction rule is no longer applicable when the set of vertices from $\hat B$ that is not saturated by $\hat M$ is empty. In this case, we have $|B| = |B \setminus \hat B| + | \hat B| \leq |A \setminus \hat A| + | \hat A| = |A|$ which is $\OO(k^2)$.
	Thus, we obtain a kernel with $\OO(k^2)$ vertices.
\end{enumerate}

\defparproblem{{\TPT}}{A tournament $D = (V, A)$ and an integer $k$.}{$k$.}{Are there at least $k$ pairwise vertex disjoint directed triangles in $D$?}

Fomin et al. \cite{FominLLSTZ19} provided a kernel with $\OO(k^{3/2})$ vertices that use both Lemma \ref{lem:expansion-lemma} (expansion lemma) and Lemma \ref{lemma:new-expansion-lemma} (stronger expansion lemma).

The authors also went on to prove Lemma \ref{lemma:double-expansion-lemma} stated below that simultaneously provides expansion in multiple graphs.

\begin{lemma}[Double Expansion Lemma \cite{FominLLSTZ19}]
\label{lemma:double-expansion-lemma}
Let $q$ be a positive integer and $G, H_{1},H_{2},\ldots,H_{d}$ be bipartite graphs with bipartitions $A \uplus B,A_{1} \uplus R_{1},\ldots,A_{d} \uplus R_{d}$ respectively, such that for every $i \neq j$, $A_{i} \cap A_{j} = \emptyset, R_{i} \cap R_{j} = \emptyset$, $\cup_{i=1}^{d} A_{i} = A$, and $\cup_{i=1}^{d} R_{i} \subseteq B$.
We can in polynomial time compute $\hat A \subseteq A, \hat B \subseteq B, \hat A_{i} \subseteq A_{i}, \hat R_{i} \subseteq R_{i}$ for every $i \in [d]$ such that
\begin{enumerate}
	\item\label{prop1:double-expansion} $\hat B = \cup_{i=1}^{d} \hat R_{i}$,
	\item\label{prop2:double-expansion} $|B \setminus \hat B| \leq q(|A| + |\cup_{i=1}^{d} \hat A_{i}|)$,
	\item\label{prop3:double-expansion} $\hat A$ has a stronger $q$-expansion into $\hat B$ and for every $i \in [d]$, $\hat A_{i}$ has a stronger $q$-expansion into $\hat R_{i}$ in $H_{i}$, 
	\item\label{prop4:double-expansion} $N_{G}(\hat B) \subseteq \hat A$, and 
	\item\label{prop5:double-expansion} for all $i \in [d]$, $N_{H_{i}}(R_{i}) \subseteq A_{i}$.
\end{enumerate}
\end{lemma}

Double expansion lemma was used in {\FVST} to give an improved kernel. We formally define the problem below.

\defparproblem{{\FVST}}{A tournament $D = (V, A)$ and an integer $k$.}{$k$.}{Is there a set $S \subseteq V(D)$ of at most $k$ vertices such that $D - S$ is acyclic?}

{\FVST} can be cast as a {\sc $3$-Hitting Set} problem, and hence its kernel can be adapted to obtain a kernel with $\OO(k^2)$ vertices.
Fomin et al.~\cite{FominLLSTZ19} used Lemma \ref{lemma:double-expansion-lemma} (double expansion lemma) to obtain a kernel with $\OO(k^{3/2})$ vertices. We give a sketch of the algorithm below.

As hitting directed cycles in a tournament is equivalent to hitting directed triangles in a tournament, a $3$-approximation algorithm follows for {\FVST}.
However, a set $S$ of size at most $7k/3$ can also be obtained due to Minch et al. \cite{MnichWV16} such that $D-S$ is acyclic.
We call an arc $(x,y) \in A$ as {\em strong arc} if $(x,y)$ is part of at least $k+2$ triangles with the other endpoints of the triangles from $S$. In this case, one of $x$ and $y$ has to be in every solution.

Since $D-S$ is acyclic, its vertices have a topological ordering.  
We can extend this to an ordering of vertices of $D$ by adding vertices in $S$ such that it minimizes certain ``conflicts". In this ordering, we can then argue that the vertices of every triangle in $D$ that do not contain a strong arc are within $7k$ distance from each other in the ordering. Hence these triangles are, in a sense {\em local}.

The problem hence boils down to simultaneously hitting strong arcs and the local triangles. This is achieved by using Lemma \ref{lemma:double-expansion-lemma} by constructing a {\em global}  graph corresponding to the strong arcs and {\em local} graphs for the local triangles.

First, the vertices of $D-S$ are broken into groups of size $100k$ according to the topological ordering.
Then, the authors \cite{FominLLSTZ19} look at the first $p=4\sqrt{k}$ such groups $Y_1, Y_2, \dotsc , Y_{p}$ and discard the first and last $7k+1$ vertices of each to form intervals $Y'_1, Y'_2, \dotsc , Y'_{p}$. For $i \in [p]$, let $S_i$ be the set of vertices in $S$ that are placed within $Y'_i$ in the ordering of $V(D)$. Note that the vertices of any local triangle containing a vertex in $Y'_i$ are contained in the set $S_i \cup Y_i$.
Then each interval $Y'_i, i \in [p]$ can be divided into $q$ subintervals $Y_{i,1}, Y_{i,2}, \dotsc , Y_{i,q}$, each of size $3\sqrt{k}$.

Next, they define a global bipartite graph $G$ with bipartition $S \uplus \{Y_{i,j} : i \in [p], j \in [q]\}$ and strong arcs.
Additionally, the authors also define local bipartite graphs $H_i$ with bipartition $S_i \uplus \{Y_{i,j} : i \in [p], j \in [q]\}$ that contains the local triangles.
Finally, they appropriately apply double expansion lemma (Lemma \ref{lemma:double-expansion-lemma}) to obtain expansion and neighborhood containment properties in all of the graphs simultaneously. This allows them to find a vertex in $V(D) \setminus S$, which can be safely deleted.
When the lemma is no longer applicable, the number of vertices is $|S|+ |\cup_{i \in [p]} Y_i| \leq 2k + 100k \cdot 4\sqrt{k}$.
This gives us a summary of the kernel with $\OO(k^{3/2})$ vertices (See \cite{FominLLSTZ19} for more details).

\subsection{Balanced Expansion and Balanced Crown Decomposition}
\label{sec:balanced-expansion}

Now, we introduce a notion of {\em balanced crown decomposition}, {\em balanced expansion}, and some related polynomial-time algorithms developed by Casel et al. \cite{Casel0INZ21}.

\begin{definition}[Balanced Expansion]
\label{defn:balanced-expansion}
Let $G = (V = A \uplus B, E, w)$ be a vertex-weighted bipartite graph, where $w^B_{\max} = \max_{b \in B} w(b)$.
For $q \in \nn_0$, a partition $A_1 \cup A_2$ of $A$, and $f: B \rightarrow A$, the tuple $(A_1, A_2, f, q)$ is called a {\em balanced expansion} if:
\begin{enumerate}
	\item $w(a) + w(f^{-1}(a)) \geq q - w_{\max}^B + 1$ if $a \in A_1$,
	\item $w(a) + w(f^{-1}(a)) \leq q + w_{\max}^B - 1$ if $a \in A_2$,
	\item for all $b \in B$, $f(b) \in N(b)$, and
	\item $N(f^{-1}(A_1)) \subseteq A_1$.
\end{enumerate}
(See Figure \ref{fig:balanced-expansion-q-3} for an illustration.)
\end{definition}

\begin{figure}[t]
\centering
	\includegraphics[scale=0.25]{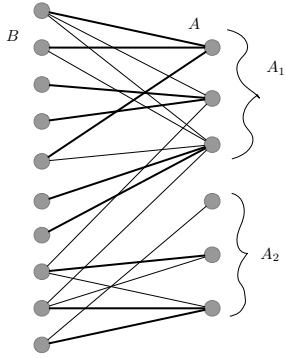}
	\caption{An illustration of a balanced expansion for $q = 3$ and with $w(v) = 1$ for every vertex $v$.  The assignment $f$ is depicted with bold edges.}
\label{fig:balanced-expansion-q-3}
\end{figure}

Based on the above definition, the weight of a vertex of $A$ combined with the weight of its preimage is in the range of $q - w^B_{\max} +1 $ and $q + w^B_{\max} - 1$. The function $f$ is similar to the condition of expansion in Section \ref{sec:expansion-lemma-basic}.
Using the above definition, Casel et al. \cite{Casel0INZ21} proved the following result.

\begin{lemma}[Balanced Expansion Lemma]
\label{lemma:balanced-expansion-lemma}
Consider a vertex-weighted bipartite graph $G = (V = A \uplus B, E, w)$ with no isolated vertices in $B$ and $q \geq \max_{b \in B} w(b) = w^B_{\max}$.
A balanced expansion $(A_1, A_2, f, q)$ for $G$ can be computed in $\OO(mn)$-time.
Furthermore, if $w(A) + w(B) \geq q|A|$, then $A_1 \neq \emptyset$.
\end{lemma}


Following is a more fine-grained version of the balanced expansion defined in \cite{Casel0INZ21}.

\begin{definition}[Fractional Balanced Expansion]
\label{defn:fractional-balanced-expansion}
Let $G = (V = A \uplus B, E, w)$ be a vertex-weighted bipartite graph, where $w^B_{\max} = \max_{b \in B} w(b)$.
For $q \in \nn_0$, a partition $A_1 \cup A_2$ of $A$, and $g: E(G) \rightarrow \nn_0$, the tuple $(A_1, A_2, g, q)$ is called a {\em fractional balanced expansion} if:
\begin{enumerate}
	\item $w(a) + \sum_{b \in B} g(ab) \geq q$ if $a \in A_1$,
	\item $w(a) + \sum_{b \in B} g(ab) \leq q$ if $a \in A_2$,
	\item for all $b \in B$, $\sum_{a \in A} g(ab) \leq w(b)$, and
	\item $N(B_U \cup B_{A_1}) \subseteq A_1$ where $B_a = \{b \in B \mid g(ab) > 0\}$ for $a \in A$, $B_{A'} = \cup_{a \in A'} B_a$ for $A' \subseteq A$ and $B_U = \{b \in B \mid \sum_{a \in A} g(ab) < w(b)\}$.
\end{enumerate}
\end{definition}

Such a tuple $(A_1, A_2, g, q)$ can be interpreted as a fractional balanced expansion as follows.
The weights of the vertices of the $f^{-1}$ function are distributed to edges through the $g$ function, i.e., $g$ can be interpreted as a fractional assignment from $v \in B$ to $u \in N(b)$.
For the non-fractional version (in Definition \ref{defn:balanced-expansion}, we aim that for all $v \in B$, $g(uv) = w(u)$ for a unique $u \in A$ while $g(u'v) = 0$ when $u' \in A \setminus \{u\}$.

The authors \cite{Casel0INZ21} prove the following fractional version of Lemma \ref{lemma:balanced-expansion-lemma}.

\begin{lemma}[Fractional Balanced Expansion Lemma]
\label{lemma:frac-balanced-expansion-lemma}
Let $G = (V= A \uplus B, E, w)$ be a vertex-weighted bipartite graph, where $w^B_{\max} = \max_{b \in B} w(b)$.
For $q \in \nn_0$, a partition $A_1 \cup A_2$ of $A$ and $g: E(G) \rightarrow \nn_0$, a fractional balanced expansion $(A_1, A_2, g, q)$ can be computed $\OO(mn)$-time.
Furthermore, if $w(A) + w(B) \geq q|A|$, then $A_1 \neq \emptyset$.
\end{lemma}

The authors~\cite{Casel0INZ21} use Lemma \ref{lemma:balanced-expansion-lemma} to prove the construction of a structure called {\em $\lambda$-balanced crown decomposition} for a fixed integer $\lambda$ described below.

\begin{definition}[$\lambda$-Balanced Crown Decomposition \cite{Casel0INZ21}] 
\label{defn:balanced-crown-decomposition}
A {\em $\lambda$-balanced crown decomposition} for a vertex-weighted graph $G = (V, E, w)$ is a tuple $(C, H, \fR, f)$ where $(H, C, R)$ is a partition of $V(G)$ and $\fR$ is a partition of $R$ and $f: \cc \cc \rightarrow H$ where $\cc \cc (C)$ is the set of connected components of $G[C]$ such that:
\begin{enumerate}
	\item there are no edges from $C$ to $R$,
	\item $w(Q) < \lambda$ for each $Q \in \cc \cc (C)$,
	\item $f(Q) \in N(Q)$ for every $Q \in \cc (C)$, 
	\item $w(h) + w(f^{-1}(h)) \geq \lambda$ for each $h \in H$, and
	\item $G[R']$ is connected and $\lambda \leq w(R') \leq 3\lambda - 3$ for each $R' \in \fR$.
\end{enumerate}
\end{definition}

The first four properties are those of crown property with some added weight specifications (see \cite{Xiao17a}), and the last property ensures that all parts of the partition are fairly balanced.
Also, the fourth condition implies that every $u \in H$ should be connected to component(s) of $G[C]$ with weight at least $\lambda - 1$.
Observe that the notion of $\lambda$-balanced crown decomposition is something that Casel et al. \cite{Casel0INZ21} have contributed towards obtaining some new results.
But a crucial contribution of Definition \ref{defn:balanced-crown-decomposition} by \cite{Casel0INZ21} is the fifth condition.
The fifth condition says that every part of the partition of $R$ is of weight within the range $[\lambda, 3\lambda - 3]$.

Lemma \ref{lemma:balanced-expansion-lemma} (balanced expansion lemma) was not directly used to prove any kernel upper bound result, but was used by Casel et al. \cite{Casel0INZ21} to prove the ``balanced crown lemma'' that we state below.

\begin{lemma}[Balanced Crown Lemma \cite{Casel0INZ21}]
\label{lemma:balanced-crown-lemma}
Let $G = (V, E, w)$ be a vertex-weighted graph and $\lambda \in \nn$ such that each connected component in $G$ has weight at least $\lambda$.
A $\lambda$-balanced crown decomposition $(C, R, \fR, f)$ of $G$ can be computed in $\OO(k^2 |V| |E|)$-time where $k = |H| + |\fR| \leq \min \{w(G)/\lambda, |V|\}$.
\end{lemma}

Casel et al. \cite{Casel0INZ21} used Lemma \ref{lemma:balanced-crown-lemma} to get a kernel for the weighted version of a separation and packing problem that we define below.
Note that {\WSeparation} can be viewed as a weighted variant of the {\sc $p$-Component Order Connectivity} problem, and {\WPacking} is simply its packing analog.

\defproblem{{\WSeparation}}{A vertex-weighted undirected graph $G = (V, E, w)$ and integer $k$.}{Is there a set $S \subseteq V(G)$ of at most $k$ vertices such that every connected component of $G - S$ has total weight less than $W$?}

\defproblem{\WPacking}{A vertex-weighted undirected graph $G = (V, E, w)$ and integers $k$.}{Are there $k$ pairwise disjoint sets $V_1,\ldots,V_k$ such that for every $1 \leq i \leq k$, $V_i \subseteq V(G), w(V_i) \geq k$ and $G[V_i]$ is connected?}

For both these problems Casel et al.\cite{Casel0INZ21} provided kernels with $3k(W - 1)$ vertices, that can 
be computed in $\OO(k^2 |V||E|)$-time.

Finally, they also studied the {\sc Balanced Connected Partition} problem stated below.

\defproblem{\BCP}{A vertex-weighted undirected graph $G = (V, E, w)$ and integer $k$.}{Is there a partition $V_1,\ldots,V_k$ of $V(G)$ such that for each $i \in [k]$, $w(V_i) \geq W$ and $G[V_i]$ is connected?}

In the same paper, they use Definition \ref{defn:balanced-crown-decomposition} and Lemma \ref{lemma:balanced-crown-lemma} to give a polynomial-time for {\BCP} that outputs a partition $V_1^*,\ldots,V_k^*$ of $V(G)$ such that for each $i \in [k]$, $G[V_i]$ is connected and $w(V_i^*) \geq W/3$.

\subsection{Additive Expansion Lemma}
\label{sec:additive-expansion}

Koana et al. \cite{KoanaNN22arxiv} introduced the following variant of expansion and expansion lemma.

\begin{definition}[$q$-Additive Expansion,~\cite{KoanaNN22arxiv}] 
\label{defn:additive-expansion}
Let $G = (V = A \uplus B, E)$ be a bipartite graph and $q$ be a positive integer.
Then, $G$ is said to have a {\em $q$-additive expansion} from $A$ into $B$ if for every $B' \subseteq B$ of size $q$, there is a matching saturating $A$ in $G[A, B \setminus B']$
\end{definition}

For completeness, we give proof of the following proposition that is stated in \cite{KoanaNN22arxiv}.

\begin{proposition}
\label{prop:additive-expansion-hall-condition}
Let $G = (V = A \uplus B, E)$ be a bipartite graph. 
There is a $q$-additive expansion from $A$ into $B$ if and only if for every  nonempty $X \subseteq A$, $|N(X)| \geq |X| + q$.
\end{proposition}

\begin{proof}
We will prove this using Hall's theorem.

Consider the forward direction ($\Rightarrow$) first.
Suppose that $G$ has a $q$-additive expansion from $A$ into $B$.
Then, for every $B' \subseteq B$ of size $q$, there is a matching saturating $A$ in $G' = G[A, B \setminus B']$.
It means that for every non-empty $X \subseteq A$, $|N(X) \setminus B'| \geq |X|$.
Consider the value of $|N_G(X)|$. 
If $|N_G(X)| < |X| + q$, then consider a set of an arbitrary set $\hat B$ of $q$ vertices from $N_G(X)$.
Clearly, $\hat B \subseteq B$ has $q$ vertices.
But, then in the graph $\hat G = G[A, B \setminus \hat B]$, for the same set $X \subseteq A$, it holds that $|N_{\hat G}(X)| < |X|$.
Then it leads to a contradiction that there is a matching saturating $A$ in $\hat G$.

Let us now prove the backward direction ($\Leftarrow$).
Suppose that for every nonempty $X \subseteq A$, $|N(X)| \geq |X| + q$.
We have to justify that for every $B' \subseteq B$ of size $q$, there is a matching saturating $A$ in $G[A, B \setminus B']$.
Suppose that for some $B^* \subseteq B$ of size $q$, there is no matching saturating $A$ in $G^* = G[A, B \setminus B^*]$.
Then, by Proposition \ref{prop:hall-theorem} (hall's theorem), there is $X \subseteq A$ such that $|N_{G^*}(X)| < |X|$.
As $B^*$ is of size $q$, then in $G$ also, it holds that $|N_G(X)| < |X| + q$.
This contradicts the premise, which completes the proof of the observation.
\end{proof}

Now, we state the {\em additive expansion lemma} and give a short explanation about how it has been used to provide an improved kernel for {\sc Partial Vertex Cover}.

\begin{lemma}[Additive Expansion Lemma \cite{KoanaNN22arxiv}]
\label{lemma:additive-expansion-lemma}
Let $q \geq 1$ and $G = (V = A \uplus B, E)$ be a bipartite graph. If $|B| > q|A|$ and there is no isolated vertex in $B$, then there exists nonempty $\hat A \subseteq A$ and $\hat B \subseteq B$ such that
\begin{itemize}
	\item there is a $q$-additive expansion of $\hat A$ into $\hat B$, and
	\item no vertex of $\hat B$ has neighbor outside $\hat A$, i.e. $N(\hat B) \subseteq \hat A$.
\end{itemize}
\end{lemma}

\defparproblem{{\sc Partial Vertex Cover}}{An undirected graph $G = (V, E)$ and two integers $k, \ell$.}{$k + \ell$}{Is there a set $S$ of at most $k$ vertices such that $G - S$ has at most $|E(G)| - \ell$ edges?}

First, the authors \cite{KoanaNN22arxiv} provide a kernel with $(\ell + 2)(k + \ell)$ vertices that uses expansion lemma (Lemma \ref{lem:expansion-lemma}). It works as follows.

\begin{itemize}
	\item First, they consider the LP formulation of {\VC}. By Nemhauser and Trotter \cite{NemhauserT75}, there is a fractional optimal solution for LP formulation of {\VC} such that all vertices are assigned values either 0, 1, or 1/2.
	\item Let $V_0, V_1, V_{1/2}$ denote the vertices that are assigned values 0, 1, 1/2, respectively. If $|V_1| + |V_{1/2}| \geq k + \ell + 1$, then the instance is a no-instance.
	\item Otherwise, they construct an auxiliary bipartite graph, and apply expansion lemma to get a kernel with $(\ell + 2)(k + \ell)$ vertices.
\end{itemize}

After that, they give a proof that using additive expansion lemma (Lemma \ref{lemma:additive-expansion-lemma}), a kernel with $(\ell' + 1)(k + \ell)$ vertices can be obtained such that $\ell' = \max\{\ell, 1\}$.
The proof has similar structures, but a crucial difference is that they use Lemma \ref{lemma:additive-expansion-lemma} instead of Lemma \ref{lem:expansion-lemma}.

\section{Conclusions and Future Research}
\label{sec:conclusion}
Our aim in this survey is to bring variations and parameterized complexity applications of an important tool called expansion lemma into one place. 
Different authors have obtained variations to suit the specific problems they deal with, and so some of the applications tend to be applicable to specific kinds of problems.
The survey is non-exhaustive simply because new variations and applications are exploding at a fast pace. We hope that our survey will trigger further applications of these variations in parameterized complexity. 
It will be interesting to see more applications of recent variants such as Lemma \ref{lemma:weighted-expansion-lemma}, Lemma \ref{lemma:double-expansion-lemma}, and Lemma \ref{lemma:additive-expansion-lemma} and to find tighter bounds for the kernel results for which they are applied.
In particular, following research directions would be interesting to investigate.
\begin{itemize}
	\item There are some other well-studied vertex deletion problems in kernelization perspective, e.g., an $\OO(k^2)$ vertex kernel for {\sc Split Vertex Deletion} (see~\cite{AgrawalGJK20}) but has a kernel lower bound of $\OO(k^{2 - \varepsilon})$ edges \cite{dell2014satisfiability}. It would be interesting to see if one of these variants of expansion lemma can be useful to get a subquadratic vertex kernel for these problems.
	\item Recently, Jacob et al. \cite{JacobM020,JacobMR21} have introduced the study of vertex deletion problems to scattered graph classes. While they have proved parameterized complexity of this problem, kernelization complexity is open and unexplored. It would be nice to explore if expansion lemma or some of its variants can be useful for this.
\end{itemize}

\end{document}